\documentclass[11pt]{article}
\usepackage{}
\usepackage[letterpaper,hmargin=1.25in,vmargin=1.25in]{geometry}
\usepackage{url}
\usepackage{cite}
\usepackage{color}
\usepackage{mathrsfs}
\usepackage{amsfonts}
\usepackage{amsmath,amsthm,amssymb}
\usepackage{cases}

\newtheorem{theorem}{Theorem}

\newtheorem{definition}{Definition}

\numberwithin{equation}{section}

\begin{document}
\title{More new classes of permutation trinomials over $\mathbb{F}_{2^n}$ }
\author{Yanping Wang
\\ISN Laboratory, Xidian University, Xi'an 710071, China\\
e-mail: ypwang@aliyun.com
\and WeiGuo Zhang
\\ISN Laboratory, Xidian University, Xi'an 710071, China\\
e-mail: zwg@xidian.edu.cn
\and Zhengbang Zha \\
School of Mathematical Sciences, Luoyang Normal University, Luoyang 471934, China\\
e-mail: zhazhengbang@163.com
}
\date{}

\maketitle

\begin{abstract}
Permutation polynomials over finite fields have wide applications in many areas of science and engineering. In this paper, we present six new classes of permutation trinomials over $\mathbb{F}_{2^n}$ which have explicit forms by determining the solutions of some equations.
\end{abstract}


\section{Introduction} Let $q$ be a power of prime $p$ and $\mathbb{F}_q$ denote the finite field with $q$ elements. Define $\mathbb{F}^{*}_q$ be the multiplication group of $\mathbb{F}_q$. A polynomial $f(x)\in \mathbb{F}_q[x]$ is called a \emph{permutation polynomial} over $\mathbb{F}_q$ if the associated polynomial function $f:c \rightarrow f(c)$ from $\mathbb{F}_q$ into $\mathbb{F}_q$ is a permutation of $\mathbb{F}_q$ \cite{LN97}. Permutation polynomials over finite fields have wide applications in combinational designs \cite{LN97}, cryptography\cite{LM1984,G1993}, and coding theory \cite{Laigle-Chapuy2007}. 
Finding new permutation polynomials is of great interest in both theoretical and applied aspects.
There are numerous interesting results about the study of permutation polynomials (see \cite{HX2005,DXY2009,ZH2012,Wang2013,PGV2014,TZH2014,ZHC2015,ZTT2015,ZYP2016,ZC2016}). 
The reader may refer to \cite[Chapter $7$]{LN97} and \cite[Chapter $8$]{GD2013} for detailed knowledge on permutation polynomials.

Permutation polynomials with few terms over finite fields are in particularly interesting for their simple algebraic forms and important applications in the areas of mathematics and engineering. Recently, the studying of permutation trinomials is a hot topic.
Some excellent works were done and numerous beautiful permutation trinomials were discovered.
For instance, Dobbertin proposed one permutation trinomial $$x^{2^{k+1}+1} + x^{3} + x$$ over $\mathbb{F}_{2^{2k+1}}$, which is applied to show that the Welch power function $x^{2^{k}+3}$ is almost perfect nonlinear(APN) \cite{dobbertin1999}. Afterwards, Dobbertin developed the multivariate method \cite{DOBBERTIN2002} to confirm the permutation property of certain types of polynomials over $\mathbb{F}_{2^{n}}$.
Hou \cite{Hou2015} determined the permutation behaviors of trinomials with the form $$ ax + bx^{q}+ x^{2q-1} $$ 
over $\mathbb{F}_{q^{2}}$ for some explicit conditions on $a$ and $b$. Moreover, Hou did a survey of permutation binomials and trinomials \cite{houxd2015}.
Lee et al. \cite{JY1997} presented some permutation trinomials of the form $x^{r}f(x^{s})$ over finite fields. Ding et al. \cite{CLQ2015} gave more classes of permutation trinomials with nonzero coefficients $1$ over finite fields with even characteristic. Motivated by \cite{CLQ2015}, Li et al. \cite{LQC2017,LQCL2017} also proposed several classes of permutation trinomials over $\mathbb{F}_{2^n}$. Bhattacharya et al. \cite{SS2016} characterized permutation trinomials of the form $$x^{2^{s}+1} + x^{2^{s-1}+1} +\alpha x$$ over $\mathbb{F}_{2^n}$ for some $\alpha\in\mathbb{F}_{2^n}$. After that, three classes of permutation trinomials were proposed by Ma et al. \cite{JTTG2016}. Gupta et al. \cite{RR2016} built four new classes of permutation trinomials of the form $x^{r}h(x^{2^{m}-1})$ over $\mathbb{F}_{2^{2m}}$. Followed the work of \cite{RR2016}, Zha et al. \cite{ZHF2017} further presented six classes of permutation trinomials of the form $x^{r}h(x^{2^{m}-1})$ over $\mathbb{F}_{2^{2m}}$.  Li et al. \cite{LT2016} constructed four new classes of permutation trinomials over $\mathbb{F}_{2^{2m}}$ from Niho exponents of the form $$x^{s(2^{m}-1)+1} + x^{t(2^{m}-1)+1} + x.$$  M.E. Zieve \cite{M2013} constructed some classes of permutation trinomials over $\mathbb{F}_{q^{2}}$ by exhibiting classes of low-degree rational functions over $\mathbb{F}_{q^{2}}$ which induce bijections on the set of $(q + 1)$-th roots of unity. Based on the paper of M.E. Zieve \cite{M2013}, Wu et al. \cite{DYDM2016} gave the explicit form of all permutation trinomials over $\mathbb{F}_{2^{m}}$, and presented some classes permutation trinomials of the form $x^{r}f(x^{s})$ over $\mathbb{F}_{2^{2m}}$.

In this paper, six new classes of permutation trinomials are proposed by determining the solutions of some equations.
The paper is organized as follows. 
In Section \ref{two}, we present five new classes of permutation trinomials over $\mathbb{F}_{2^n}$ by applying the multivariate method.
In Section \ref{three}, we obtain a new class of permutation trinomials over $\mathbb{F}_{2^n}$ for the case of $n\equiv 0({\rm mod}\ 4)$. Section \ref{conclusion} is the conclusion of this paper.

\section{Five new classes of permutation trinomials}\label{two}
Inspired by the idea of \cite{DOBBERTIN2002}, we present five new classes of permutation trinomials in this section by using the multivariate method.

\subsection{The case of $n\equiv 0({\rm mod}\ 3)$}
In this subsection, we propose two classes of permutation trinomials over $\mathbb{F}_{2^{n}}$ when $n\equiv 0({\rm mod}\ 3)$.
First we give the definition of a trace function needed later.

\begin{definition}\label{defi:1} Let $n$ and $k$ be positive integers with $k\mid n$. The trace function ${\rm Tr}^{n}_{k}(x)$ from $\mathbb{F}_{p^n}$ to $\mathbb{F}_{p^k}$ is defined by
$${\rm Tr}^{n}_{k}(x)= x+ x^{p^{k}}+ x^{p^{2k}}+ \cdots + x^{p^{(\frac{n}{k}-1)k}}.$$
If $k=1$, then ${\rm Tr}^{n}_{1}(x)$ is called the absolute trace function.
\end{definition}

\begin{theorem}\label{th:K2}
Let $k\not\equiv 2~({\rm mod}\ 3)$ and $n=3k$. Then
\begin{eqnarray*}
f(x)=x^{2^{2k}+2^{k}-1} + x^{2^{2k}} +x
\end{eqnarray*}
is a permutation polynomial on $\mathbb{F}_{2^n}$.
\end{theorem}

\begin{proof} We shall show that for each $a \in\mathbb{F}_{2^n}$, the equation
\begin{eqnarray}\label{eq:K2.1}
f(x)=x^{2^{2k}+2^{k}-1} + x^{2^{2k}} +x =a
\end{eqnarray}
has one unique solution in $\mathbb{F}_{2^n}$.

When $a=0$, we need to prove that $f(x)=0$ has the unique solution $x=0$. Assume $x\neq0$ is a solution of the equation
\begin{eqnarray}\label{eq:K2.2}
x^{2^{2k}+2^{k}-1} + x^{2^{2k}} + x =0.
\end{eqnarray}
Raising both sides of (\ref{eq:K2.2}) to the $2^{k}$-th power, we have
\begin{eqnarray*}
x^{2^{2k}-2^{k}+1} + x^{2^{k}} + x =0,
\end{eqnarray*}
which means that
\begin{eqnarray}\label{eq:K2.3}
x^{2^{2k}-2^{k}} + x^{2^{k}-1} + 1 =0.
\end{eqnarray}
Setting $y=x^{2^{k}-1}$, equation (\ref{eq:K2.3}) can be written as
\begin{eqnarray}\label{eq:K2.3'}
y^{2^{k}} + y + 1 =0.
\end{eqnarray}
Applying the mapping ${\rm Tr}^{n}_{k}(\cdot)$ on both sides of equation (\ref{eq:K2.3'}), we get ${\rm Tr}^{n}_{k}(1)=0$, which contradicts with ${\rm Tr}^{n}_{k}(1)=1$. Hence, equation (\ref{eq:K2.3}) has no solution.

When $a \neq 0$, we will show that equation (\ref{eq:K2.1}) has one non-zero solution.
Let $y=x^{2^{k}}$, $z=y^{2^{k}}$, $b=a^{2^{k}}$ and $c=b^{2^{k}}$. Then we can obtain the system of equations
\begin{equation}\label{eq:K2.4}
\left\{
\begin{array}{l}
      x + z +  \frac{yz}{x} =a,  \\
     y + x +  \frac{xz}{y} =b,   \\
     z + y +  \frac{xy}{z} =c.
\end{array}
\right.
\end{equation}
Let $u=\frac{yz}{x}$, $v=\frac{xz}{y}$, and $w=\frac{xy}{z}$.  Then $x^{2}=vw$, $y^{2}=uw$, and $z^{2}=uv$. We have
\begin{equation}\label{eq:K2.5}
\left\{
    \begin{array}{lll}
      u^{2} + vw +  uv =a^{2},  \\
     v^{2} + vw +  uw =b^{2},   \\
     w^{2} + uw +  uv =c^{2}.
    \end{array}
  \right.
\end{equation}
By (\ref{eq:K2.5}), we deduce that $u^{2}+v^{2}+w^{2}=a^{2}+b^{2}+c^{2}$, which implies $u +v +w =a +b +c$. Setting $\epsilon=a +b +c$,  we have $\epsilon\in\mathbb{F}_{2^k}$ and $w=\epsilon +u +v$.

If $\epsilon=0$, then pluging $u +v +w =0$ into (\ref{eq:K2.5}), we have $u =b$, $v =c$, $w =a$. So the equation has a unique solution $x=(vw)^{\frac{1}{2}}=(ac)^{\frac{1}{2}}$.

If $\epsilon \neq0$, then from (\ref{eq:K2.5}), we can get
\begin{equation}\label{eq:K2.6}
\left\{
    \begin{array}{ll}
      u^{2} + v^{2} +  \epsilon v = a^{2},  \\
     u^{2} + \epsilon u +  \epsilon v = b^{2}.
    \end{array}
  \right.
\end{equation}
It can be verified that $u=v^{2^{2k}}$ and
\begin{eqnarray}\label{eq:K2.7}
v^{2^{2k+1}} + v^2 +  \epsilon v = a^{2}.
\end{eqnarray}
Eliminating the indeterminate $u$ in (\ref{eq:K2.6}), we obtain
\begin{eqnarray}\label{eq:K2.8}
v^{4} + \epsilon^{2}v^{2} + \epsilon^{3}v + a^{4}+b^{4}+a^{2}\epsilon^{2} =0.
\end{eqnarray}
Since $\epsilon\in\mathbb{F}_{2^k}$ and $\epsilon\neq0$, substituting $v$ with $\epsilon v$ in (\ref{eq:K2.7}) and (\ref{eq:K2.8}), we obtain
\begin{equation}\label{eq:K2.9}
\left\{
    \begin{array}{ll}
     v^{2^{2k+1}} +  v^2 +  v = a^{2}/\epsilon^2,  \\
     v^{4} + v^{2} + v + \frac{1}{\epsilon^{4}}(a^{4}+b^{4}+a^{2}\epsilon^{2}) =0.
    \end{array}
  \right.
\end{equation}
Suppose $v_1$ and $v_2$ are the solutions of (\ref{eq:K2.9}) and let $\lambda=v_1+v_2$. Then we get $\lambda=0$ or
\begin{equation}\label{eq:K2.10}
\left\{
    \begin{array}{ll}
     \lambda^{2^{2k+1}-1} + \lambda +  1 = 0,  \\
     \lambda^{3} + \lambda+1 =0,
    \end{array}
  \right.
\end{equation}
which means that $\lambda^{2^{2k+1}-4}=1$. Since $k\not\equiv 2 ~({\rm mod}\ 3)$, we can deduce that
$$\gcd(2^{2k+1}-4,2^{3k}-1)=2^{\gcd(2k-1,3k)}-1=1$$
and $\lambda=1$. Note that $\lambda=1$ is not the solution of (\ref{eq:K2.10}). Therefore, $\lambda=0$ and there is
a unique solution $v$ of (\ref{eq:K2.9}). Hence, (\ref{eq:K2.4}) has a unique solution when $k\not\equiv 2({\rm mod}\ 3)$.
\end{proof}

\begin{theorem}\label{th:N2}
Let $k\not\equiv 2~({\rm mod}\ 3)$ and $n=3k$. Then
\begin{eqnarray*}
f(x)=x^{2^{2k}+2^{k}-1} + x^{2^{2k}-2^{k}+1} +x
\end{eqnarray*}
is a permutation polynomial on $\mathbb{F}_{2^n}$.
\end{theorem}

\begin{proof} For any fixed $a \in\mathbb{F}_{2^n}$, we need to prove that the equation
\begin{eqnarray}\label{eq:K4.1}
f(x)=x^{2^{2k}+2^{k}-1} + x^{2^{2k}-2^{k}+1} +x =a
\end{eqnarray}
has one unique solution in $\mathbb{F}_{2^n}$. Obviously, $f(0)=0$. For $a=0$, we have $x=0$ or
\begin{eqnarray}\label{eq:K4.2}
x^{2^{2k}+2^{k}-2} + x^{2^{2k}-2^{k}} + 1=0.
\end{eqnarray}
Setting $y=x^{2^{k}-1}$, equation (\ref{eq:K4.2}) can be written as
\begin{eqnarray}\label{eq:K4.3}
y^{2^{k}+2} + y^{2^{k}}  + 1 =0.
\end{eqnarray}
Note that $y\neq0$ and $y^{2^{2k}+2^k+1}=1$. Equation (\ref{eq:K4.3}) leads to
$$y^{2^{k}+2} + y^{2^{k}}  + y^{2^{2k}+2^k+1} =0,$$
which implies that
$$y^2+1+y^{2^{2k}+1} =0.$$
Raising both sides of the above equation to the $2^{k-1}$-th power gives
\begin{eqnarray}\label{eq:K4.4}
y^{2^{k}}  + 1 + y^{2^{k-1}(2^{2k}+1)} =0.
\end{eqnarray}
From equations (\ref{eq:K4.3}) and (\ref{eq:K4.4}), we obtain $y^{2^{k}+2}=y^{2^{k-1}(2^{2k}+1)}$, which implies that $y^{2^k+3}=1$.
Since $k\not\equiv 2~({\rm mod}\ 3)$, it can be checked that
$$\gcd(2^{2k}+2^k+1,2^k+3)=\gcd(2^k+3,7)=1.$$
Then we get $y=1$. But $y=1$ is not the solution of (\ref{eq:K4.3}). Thus, $f(x)=0$ if and only if $x=0$.

When $a \in\mathbb{F}^{*}_{2^n}$.
Let $y=x^{2^{k}}$, $z=y^{2^{k}}$, $b=a^{2^{k}}$ and $c=b^{2^{k}}$. Then we obtain the following system of equations
\begin{numcases}{}
  \frac{yz}{x}+ \frac{xz}{y} + x=a,\label{eq:K4.8}\\
  \frac{xz}{y} + \frac{xy}{z} + y=b, \label{eq:K4.8'}\\
  \frac{xy}{z} + \frac{yz}{x} + z =c. \label{eq:K4.8"}
\end{numcases}
Adding the three equations above, we obtain $x+y+z=a+b+c$. 
Equations (\ref{eq:K4.8}) and (\ref{eq:K4.8'}) lead to
\begin{numcases}{}
      y^{2}z + x^{2}z +  x^{2}y = axy, \label{eq:K4.9}\\
     xz^{2} + xy^{2} +  y^{2}z = byz .  \label{eq:K4.9'}
\end{numcases}
Adding (\ref{eq:K4.9}) with (\ref{eq:K4.9'}) gives
\begin{equation}\label{eq:K4.10}
x(y+z)\epsilon + y(ax+bz) =0,
\end{equation}
where $\epsilon=a +b +c$ and $\epsilon\in\mathbb{F}_{2^k}$.
Plugging $x=\epsilon + y + z$ into (\ref{eq:K4.10}), we get
\begin{equation}\label{eq:K4.11}
(\epsilon+a)y^{2} + \epsilon z^{2} +(a+b)yz + \epsilon(\epsilon+a)y + \epsilon^{2}z =0.
\end{equation}
Setting $y=\lambda z$ with $\lambda\neq0$, and plugging it into (\ref{eq:K4.11}), we obtain
\begin{equation}\label{eq:K4.12}
(\epsilon+a)\lambda^{2}z^{2} + \epsilon z^{2} +(a+b)\lambda z^{2} + \epsilon(\epsilon+a)\lambda z + \epsilon^{2}z =0,
\end{equation}
which can be rewritten as
\begin{equation}\label{eq:K4.13}
[(\epsilon+a)\lambda ^{2} + (a+b)\lambda + \epsilon]z^{2} + [\epsilon(\epsilon+a)\lambda + \epsilon^{2}]z =0.
\end{equation}
Let $\xi_{1}=(\epsilon+a)\lambda^{2} + (a+b)\lambda + \epsilon$ and $\xi_{2}=\epsilon(\epsilon+a)\lambda + \epsilon^{2}$. Then we have
\begin{equation}\label{eq:K4.14}
\xi_{1} z + \xi_{2} =0.
\end{equation}
Multiplying equation (\ref{eq:K4.8'}) by $yz$ and then plugging $x=\epsilon + y + z$, it leads to
\begin{equation}\label{eq:K4.15}
y^{3}+ z^{3} +yz^{2}+ \epsilon(y^{2} + z^{2})+ byz =0.
\end{equation}
Plugging $y=\lambda z$ into (\ref{eq:K4.15}) gives
\begin{equation}\label{eq:K4.16}
(\lambda^{3}+ \lambda+ 1)z + (\epsilon \lambda^{2}+ b\lambda + \epsilon) =0.
\end{equation}
Let $\eta_{1}=\lambda^{3}+ \lambda+ 1$ and $\eta_{2}= \epsilon \lambda^{2}+ b\lambda +\epsilon$, equation (\ref{eq:K4.16}) becomes
\begin{equation}\label{eq:K4.17}
\eta_{1} z + \eta_{2} =0.
\end{equation}
By equations (\ref{eq:K4.14}) and (\ref{eq:K4.17}), we can deduce that
$$\xi_{2}\eta_{1} + \xi_{1}\eta_{2}=0.$$
Then
\begin{equation}\label{eq:K4.18}
[(\epsilon+a)\lambda^{2} + (a+b)\lambda + \epsilon](\epsilon \lambda^{2}+ b\lambda +\epsilon) + [\epsilon(\epsilon+a)\lambda + \epsilon^{2}](\lambda^{3}+ \lambda +1) =0.
\end{equation}
We have
\begin{equation}\label{eq:K4.19}
(\epsilon^{2} + ab + a\epsilon) \lambda^{3}+ (\epsilon^{2}+ab+ b^{2})\lambda^{2} =0,
\end{equation}
which implies
\begin{equation}\label{eq:K4.20}
\zeta_{1}\lambda+ \zeta_{2} =0,
\end{equation}
where $ \zeta_{1}= \epsilon^{2} + ab + a\epsilon=ac+b^2+c^2$ and $\zeta_{2}=\epsilon^{2}+ab+ b^{2}=ab+a^2+c^2$.

Next, we verify $ \zeta_{1}\neq 0$. Suppose that $ \zeta_{1}=0$. Then
$$a^{1+2^{2k}}+ a^{2^{k+1}}+ a^{2^{2k+1}} =0.$$
Since $a\neq0$, the above equation leads to
$$a^{2^{k+1}(2^{k}-1)}+ a^{(2^k-1)(2^k-1)}+ 1 =0.$$
Setting $\mu=a^{2^{k}-1}$, we obtain $\mu^{2^{2k}+2^k+1}=1$ and
\begin{equation}\label{eq:K4.21}
\mu^{2^{k+1}}+ \mu^{2^{k}-1}+ 1 =0.
\end{equation}
Note that
$$\zeta_{1}+ \zeta_{1}^{2^{k}}+ \zeta_{1}^{2^{2k}}=ab+ac+bc=0,$$
which implies that
\begin{equation}\label{eq:K4.22}
a^{2^{2k}+2^{k}}+ a^{2^{2k}+1}+ a^{2^{k}+1} =0.
\end{equation}
It leads to
\begin{equation}\label{eq:K4.23}
a^{2^{2k}-1} + a^{2^{2k}-2^{k}}+ 1 =0,
\end{equation}
which can be written as
\begin{equation}\label{eq:K4.25}
\mu^{2^{k}+1}+ \mu^{2^{k}}+ 1 =0.
\end{equation}
Multiplying equation (\ref{eq:K4.25}) by $\mu^{2^k-1}$,  and then adding equation (\ref{eq:K4.21}), we get
\begin{equation}\label{eq:K4.26}
\mu^{2^{k+1}-1} +1=0.
\end{equation}
It can be verified that
$$
\begin{array}{rl}
&\gcd(2^{2k}+2^k+1,2^{k+1}-1)=\gcd(2^{k+1}-1,2^k+2^{k-1}+1)\\
=&\gcd(3\cdot 2^{k-1}+1,2^k+3)=\gcd(3\cdot 2^{k}+2,2^k+3)=\gcd(7,2^k+3)=1.
\end{array}
$$
Then equation (\ref{eq:K4.26}) has a unique solution $\mu=1$. However, $\mu=1$ is not the solution of (\ref{eq:K4.21}).
Therefore, we get $ \zeta_{1}\neq 0$ and $$\lambda=\frac{\zeta_{2}}{\zeta_{1}}=\zeta_{1}^{2^k-1}=(ac+b^2+c^2)^{2^k-1},$$
which implies that $\lambda^{2^{2k}+2^k+1}=1$.

When $k\equiv 0~({\rm mod}\ 3)$, we can deduce that $\gcd(7,2^{2k}+2^k+1)=1$. If $\lambda^{3}+ \lambda +1=0$, then $\lambda^{7}=1$, which means that $\lambda=1$.
It leads to $\lambda^{3}+ \lambda +1=1$, which is a contradiction. Therefore, $\lambda^{3}+ \lambda +1\neq0$.
By (\ref{eq:K4.16}), we get a unique solution $x=z^{2^{k}}=\big(\frac{\epsilon(1+\lambda^{2})+b\lambda}{\lambda^{3}+ \lambda +1} \big)^{2^{k}}$ of $f(x)=a$ in (\ref{eq:K4.1}).

When $k\equiv 1~({\rm mod}\ 3)$, it can be verified that $2^{2k}\equiv 4({\rm mod}\ 7)$.
If $\lambda^{3}+ \lambda +1=0$, then $\lambda^{7}=1$, $z=\lambda^{-1}y=\lambda^{6}y$ and
$$x=y^{2^{2k}}= \lambda^{2^{2k}}z^{2^{2k}}=\lambda^{4}y.$$
Plugging them into equation (\ref{eq:K4.8'}), we obtain
$$(\lambda^{10}+ \lambda^{-2}+ 1)y = b,$$
which is equivalent to
\begin{equation}\label{eq:K4.27}
(\lambda^{5}+ \lambda^{3}+1)y = b.
\end{equation}
Supposing that $\lambda^{5}+ \lambda^{3}+1=0$, we get $\lambda^{5}+\lambda=0$, which implies that $\lambda=1$. However, $\lambda=1$ is not the solution of $\lambda^{5}+ \lambda^{3}+1=0$.  So $\lambda^{5}+ \lambda^{3}+1\neq0$.
Hence, $y=\frac{b}{\lambda^{5}+ \lambda^{3}+1}$,  thereby $x=\frac{b\lambda^{4}}{\lambda^{5}+ \lambda^{3}+1}$.
If $\lambda^{3}+ \lambda +1\neq0$, we can  similarly obtain one solution $x=z^{2^{k}}=\big(\frac{\epsilon(1+\lambda^{2})+b\lambda}{\lambda^{3}+ \lambda +1} \big)^{2^{k}}$ of (\ref{eq:K4.1}) from (\ref{eq:K4.16}).
This completes the proof.
\end{proof}

\subsection{The case of $n\equiv 1({\rm mod}\ 3)$}
In this subsection, we give one class of permutation trinomials over $\mathbb{F}_{2^{n}}$ when $n\equiv 1({\rm mod}\ 3)$.

\begin{theorem}
Let $k$ be a positive integer and $n=3k+1$. Then 
\begin{eqnarray*}
f(x)=x^{2^{2k+1}+2^{k+1}+1} + x^{2^{k+1}+1} + x
\end{eqnarray*}
is a permutation polynomial on $\mathbb{F}_{2^n}$.
\end{theorem}

\begin{proof}  We prove that for any fixed $a \in\mathbb{F}_{2^n}$, the equation
\begin{eqnarray}\label{eq: d2.1}
x^{2^{2k+1}+2^{k+1}+1} + x^{2^{k+1}+1} + x =a
\end{eqnarray}
has one unique solution in $\mathbb{F}_{2^n}$.
It is obvious that $f(0)=0$. For $a=0$, we have $x=0$ or
\begin{eqnarray}\label{eq:d2.3}
x^{2^{2k+1}+2^{k+1}} + x^{2^{k+1}} + 1 =0.
\end{eqnarray}
Let $y=x^{2^{k}}$, $z=y^{2^{k}}$, then $x=z^{2^{k+1}}$. Equation (\ref{eq:d2.3}) can be written as
\begin{eqnarray}\label{eq:d2.4}
y^{2}z^{2} + y^{2} + 1=0.
\end{eqnarray}
Raising both sides of equation (\ref{eq:d2.4}) to $2^{k}$-th power, we get
\begin{eqnarray}\label{eq:d2.5}
xz^{2} + z^{2} + 1 =0.
\end{eqnarray}
Further, raising both sides of equation (\ref{eq:d2.5}) to the $2^{k}$-th power gives
\begin{eqnarray}\label{eq:d2.6}
xy + x + 1 =0.
\end{eqnarray}
Combining equations (\ref{eq:d2.3}), (\ref{eq:d2.4}) and  (\ref{eq:d2.5}) eliminates the indeterminate $y$ and $z$, we obtain $x=1$.
However $f(1)=1$. So equation (\ref{eq:d2.3}) has no solution. Hence $f(x)=0$ only has a solution $x=0$.

When $a\neq0$, we prove that $f(x)=a$ has one unique solution. Let $y=x^{2^{k}}$, $z=y^{2^{k}}$, $b=a^{2^{k}}$ and $c=b^{2^{k}}$,
then $x=z^{2^{k+1}}$ and $a=c^{2^{k+1}}$. Therefore, we obtain
\begin{numcases}{} 
       xy^{2}z^{2}+xy^{2}+x =a,   \label{eq:d2.7} \\
       xyz^{2}+yz^{2}+y =b,   \label{eq:d2.7'} \\
       xyz+xz+ z =c.  \label{eq:d2.7''}
\end{numcases}
Combining equations (\ref{eq:d2.7}) and (\ref{eq:d2.7'}) eliminates the indeterminate $z$, we have
\begin{eqnarray}\label{eq:d2.8}
  x^{2}y^{2}+ bxy +x^{2} + ax+x+a=0.
\end{eqnarray}
Adding equations (\ref{eq:d2.7}) and (\ref{eq:d2.7'}) multiply by $y$ together results in
\begin{eqnarray}\label{eq:d2.9}
xy^{2}+ y^{2}z^{2} + y^{2} + x + by + a =0.
\end{eqnarray}
Computing $(\ref{eq:d2.7})+ (\ref{eq:d2.7'})+(\ref{eq:d2.7''})*yz$, there is
\begin{eqnarray}\label{eq:d2.9}
xy^{2}+ cyz +x + y + a + b =0.
\end{eqnarray}
By (\ref{eq:d2.7''})$*cy+$(\ref{eq:d2.9})$*(xy + x+1)$, this gives
\begin{eqnarray}\label{eq:d3.0}
x^{2}y^{3}+ x^{2}y^{2}+ x^{2}y +(a+b+1)xy +x^{2} + (a+b+1)x+ (c^{2}+1)y +(a+b)=0.
\end{eqnarray}
Multiplying equation (\ref{eq:d2.8}) by $y$ and then adding (\ref{eq:d3.0}), it leads to
\begin{eqnarray}\label{eq:d3.1}
 x^{2}y^{2}+ bxy^{2} + bxy +x^{2} + (a+b+1)x+ (a+c^{2}+1)y +(a+b)=0.
\end{eqnarray}
Adding (\ref{eq:d2.8}) and (\ref{eq:d3.1}) derives
\begin{eqnarray}\label{eq:d3.2}
bxy^{2} + (a+c^{2}+1)y +(bx+b) =0.
\end{eqnarray}
Further, multiplying equation (\ref{eq:d2.9}) by $b$ and then adding (\ref{eq:d3.2}), we deduce
\begin{eqnarray}\label{eq:d3.3}
bcyz + (a+b+c^{2}+1)y + (ab + b^{2} +b) =0.
\end{eqnarray}
Raising both sides of equation (\ref{eq:d3.3}) to the $2^{2k+1}$-th power results in
\begin{eqnarray}\label{eq:d3.4}
abxy +(a+b^{2}+c^{2}+1)x + (ac^{2}+a^{2}+a) =0.
\end{eqnarray}
Multiply equation (\ref{eq:d2.8}) by $a^{2}b^{2}$,  and then adding the square of $(\ref{eq:d3.4})$, we obtain
\begin{eqnarray}\label{eq:d3.5}
a^{2}b^{3}xy +  (a^{2}+b^{4}+c^{4}+1)x^{2} + a^{2}b^{2}(x^{2} + ax +x +a)+(a^{2}c^{4}+a^{4}+a^{2}) =0.
\end{eqnarray}
Furthermore, multiplying equation (\ref{eq:d3.4}) by $ab^{2}$ and then adding (\ref{eq:d3.5}), there is
\begin{eqnarray}\label{eq:d3.6}
(a^{2}+a^{2}b^{2}+b^{4}+c^{4}+1)x^{2} + ab^{2}(a^{2}+b^{2}+c^{2}+1)x + a^{2}(a^{2}+b^{2}+b^{2}c^{2}+c^{4}+1)=0.
\end{eqnarray}
Equation (\ref{eq:d3.6}) can rewrite as
\begin{eqnarray}\label{eq:d3.7}
(a^{2}+a^{2}b^{2}+b^{4}+c^{4}+1)(x^{2}+a^{2}) + ab^{2}(a^{2}+b^{2}+c^{2}+1)(x + a)=0.
\end{eqnarray}
Supposing
\begin{eqnarray}\label{eq:d3.8}
a^{2}+a^{2}b^{2}+b^{4}+c^{4}+1=0,
\end{eqnarray}
raising both sides of equation (\ref{eq:d3.8}) to the $2^{k}$-th power, this leads to
\begin{eqnarray}\label{eq:d3.9}
a^{2}+b^{2}+b^{2}c^{2}+c^{4}+1=0.
\end{eqnarray}
Adding equations (\ref{eq:d3.8}) and (\ref{eq:d3.9}), we get
\begin{eqnarray}\label{eq:d3.10}
a^{2}+b^{2}+c^{2}+1=0.
\end{eqnarray}
Again raising both sides of equation (\ref{eq:d3.10}) to the $2^{k}$-th power, and then adding equation (\ref{eq:d3.10}). There is
$$a^{2}+a=0,$$
which means $a=0$ or $a=1$. Since $a\neq0$, we obtain $b=c=1$ from $a=1$. Thereby $a^{2}+a^{2}b^{2}+b^{4}+c^{4}+1\neq0$, this is a contradiction.
Hence, there are two solutions of equation (\ref{eq:d3.7}), which is $x=a$ or
$$x=a+\frac{a^{2}+b^{2}+c^{2}+1}{a^{2}+a^{2}b^{2}+b^{4}+c^{4}+1}.$$

If $x=a$, then $y=b$ and we deduce
$ab^{2}(a+1)=0$ from equation (\ref{eq:d2.8}). Because $a\neq0$, this leads to $b\neq0$, and hence $a=1$. Therefore equation has one unique solution $x=1$.

If $x\neq a$, then there is one unique solution $x=a+\frac{a^{2}+b^{2}+c^{2}+1}{a^{2}+a^{2}b^{2}+b^{4}+c^{4}+1}$.
The proof finishes.
\end{proof}

\subsection{The case of $n\equiv 2({\rm mod}\ 3)$}
In this subsection, two classes of permutation trinomials are proposed over $\mathbb{F}_{2^{n}}$ when $n\equiv 2({\rm mod}\ 3)$.

\begin{theorem}
Let $k$ be a positive integer and $n=3k-1$. Then 
\begin{eqnarray*}
f(x)=x^{2^{3k-1}-2^{2k}+2^{k}} + x^{2^{k}-1} + x
\end{eqnarray*}
is a permutation polynomial on $\mathbb{F}_{2^n}$.
\end{theorem}

\begin{proof} We shall prove that for each fixed $a \in\mathbb{F}_{2^n}$, the equation
\begin{eqnarray}\label{eq:J3.1}
f(x)= x^{2^{3k-1}-2^{2k}+2^{k}} + x^{2^{k}-1} + x =a
\end{eqnarray}
has one unique solution in $\mathbb{F}_{2^n}$.

When $a=0$, we claim that $x=0$. Otherwise $x\neq0$ is a solution of the equation
\begin{eqnarray}\label{eq:J3.2}
x^{2^{3k-1}-2^{2k}+2^{k}} + x^{2^{k}-1} + x =0.
\end{eqnarray}
This implies that equation
\begin{eqnarray}\label{eq:J3.3}
x^{2^{3k-1}-2^{2k}+2^{k}-1} + x^{2^{k}-2} + 1 =0
\end{eqnarray}
has nonzero solution. Let $y=x^{2^{k}}$, $z=y^{2^{k}}$, then $x^{2}=z^{2^{k}}$. Hence there is
\begin{numcases}{}
      \frac{y}{z} + \frac{y}{x^{2}} + 1 =0,  \label{eq:J3.4} \\
      \frac{z}{x^{2}} + \frac{z}{y^{2}} + 1 =0. \label{eq:J3.4'}
\end{numcases}
Combining equations (\ref{eq:J3.4}) and (\ref{eq:J3.4'}) to eliminate $z$, we obtain
\begin{eqnarray}\label{eq:J3.5}
x^{4}y(y + 1) =0.
\end{eqnarray}
Since $x\neq0$, this means $y\neq0$. Thereby we get $y=1$. It leads to $x=1$ for
$$\gcd(2^{k}-1, 2^{n}-1)=2^{(k, 3k-1)}-1=1.$$
However $f(1)=1$, this is a contradiction. Hence $f(x)=0$ if and only if $x=0$.

When $a\neq0$, we will verify that $f(x)=a$ has one unique nonzero solution. Let $y=x^{2^{k}}$, $z=y^{2^{k}}$, $b=a^{2^{k}}$ and $c=b^{2^{k}}$,
then $x^{2}=z^{2^{k}}$ and $a^{2}=c^{2^{k}}$. Therefore, we obtain
\begin{equation}\label{eq:J3.6}
\left\{
\begin{array}{l}
  \frac{xy}{z} + \frac{y}{x} + x =a,   \\
  \frac{yz}{x^{2}} + \frac{z}{y} + y =b,  \\
  \frac{x^{2}z}{y^{2}} + \frac{x^{2}}{z} + z =c.
\end{array}
\right.
\end{equation}
Eliminating the indeterminate $z$ in (\ref{eq:J3.6}), we get
\begin{numcases}{} 
      (x^{2}+ax+b)y + (x^{2}+bx^{2}+abx)=0,   \label{eq:J3.7} \\
      (x^{2}+c+ 1)y^{2} +(cx^{2}+acx)y + a^{2}x^{2} =0. \label{eq:J3.7'}
\end{numcases}
Combining equations (\ref{eq:J3.7}) and (\ref{eq:J3.7'}) to eliminate the indeterminate $y$, we obtain
\begin{equation}\label{eq:J3.8}
\alpha x^{3}+\beta x^{2}+\gamma x+\theta =0,
\end{equation}
where $\alpha =a^{2}+b^{2}+bc+c+1$, $\beta =abc$, $\gamma =a^{4}+a^{2}bc+a^{2}b^{2}+ a^{2}c+b^{2}+bc +c+1$, $\theta =a^{3}bc+abc$.

If $\alpha =0$, then we get $\gamma =0$. From (\ref{eq:J3.8}) we can deduce that $x^{2}=a^{2}+1$, which implies $x=a+1$. Therefore equation (\ref{eq:J3.8}) has one solution.

Next, substituting $x$ with $u+\alpha^{-1}\beta$ in (\ref{eq:J3.8}), we get
\begin{equation}\label{eq:J3.9}
u^{3} + (\alpha^{-2}\beta^{2}+\alpha^{-1}\gamma)u + (\alpha^{-2}\beta\gamma+\alpha^{-1}\theta)=0.
\end{equation}
Since $\alpha\theta = \beta\gamma$, this implies $\alpha^{-2}\beta\gamma+\alpha^{-1}\theta=0$. Hence we obtain $u=0$ or $u=(\alpha^{-2}\beta^{2}+\alpha^{-1}\gamma)^{1/2}$. 
Thereby there is $x=\alpha^{-1}\beta$ or $x=(\alpha^{-2}\beta^{2}+\alpha^{-1}\gamma)^{1/2}+\alpha^{-1}\beta$.
Since $\gamma =(a^{2}+1)\alpha$, then $\alpha^{-1}\gamma =a^{2}+1$. It can be easily checked that $x=(\alpha^{-2}\beta^{2}+\alpha^{-1}\gamma)^{1/2}+\alpha^{-1}\beta=a+1$. Assume that $x=a+1$, then $y=b+1$. Plugging $y=b+1$ into equation (\ref{eq:J3.7}), by some simplifying computation we obtain
$$(a+b)(a+b+1)=0.$$

If $a=b$, we have $x=y=z=a+1$. Further we get $x=0$ from equation (\ref{eq:J3.6}), which is a contradiction.

If $a=b+1$, we deduce that $x=a+1$, $y=a$, $z=a+1$. Plugging them into equation (\ref{eq:J3.6}), there is
\begin{equation}\label{eq:J3.10}
a^{2}+a+1=0.
\end{equation}

Suppose $n$ is odd, it can be easily verified that equation (\ref{eq:J3.10}) has no solution.
Suppose $n$ is even, then $k$ is odd and $2^{k}\equiv 2({\rm mod}\ 3)$. From equation (\ref{eq:J3.10}), we obtain two solutions  $a=\omega$ and $\omega^{2}$.
If $a=\omega$, it derives that $$x=\omega^{2}, y=\omega, z=\omega^{2}.$$
If $a=\omega^{2}$, it leads to $$x=\omega, y=\omega^{2}, z=\omega.$$
Without loss of generality, we consider the case of $a=\omega$. We can deduce that
$\alpha=a^{2}+b^{2}+bc+c+1=\omega^{2}$ and $\beta =abc=1$. Thereby, we get another solution $x=\alpha^{-1}\beta=\omega^{2}$. But, it equals to the solution
$x=a+1=\omega^{2}$.

Therefore, equation (\ref{eq:J3.8}) has a unique solution. We complete the proof.
\end{proof}

\begin{theorem}
Let $k$ be a positive integer and $n=3k-1$. Then 
\begin{eqnarray*}
f(x)=x^{2^{2k}+2^{k}+1} + x^{2^{2k}+1} + x
\end{eqnarray*}
is a permutation polynomial on $\mathbb{F}_{2^n}$.
\end{theorem}

\begin{proof} For each fixed $a \in\mathbb{F}_{2^n}$, it suffices to prove that equation
\begin{eqnarray}\label{eq:D3.1}
f(x)= x^{2^{2k}+2^{k}+1} + x^{2^{2k}+1} + x =a
\end{eqnarray}
has one unique solution in $\mathbb{F}_{2^n}$.

We first verify that $f(x)=0$ if and only if $x=0$. Assume that $x\neq0$ is a solution of the equation
\begin{eqnarray}\label{eq:D3.2}
f(x)=x^{2^{2k}+2^{k}+1} + x^{2^{2k}+1} + x =0,
\end{eqnarray}
which implies that equation
\begin{eqnarray}\label{eq:D3.3}
x^{2^{2k}+2^{k}} + x^{2^{2k}} +1 =0
\end{eqnarray}
has nonzero solution. Let $y=x^{2^{k}}$, $z=y^{2^{k}}$, then $x^{2}=z^{2^{k}}$. The above equation leads to
\begin{numcases}{}
        yz + z + 1 =0,   \label{eq:D3.4} \\
        x^{2}z + x^{2} + 1 =0,   \label{eq:D3.4'} \\
        x^{2}y^{2} + y^{2} + 1=0.  \label{eq:D3.4''}
\end{numcases}
Combining (\ref{eq:D3.4}) and (\ref{eq:D3.4'}) to eliminate the indeterminate $z$, we obtain 
\begin{eqnarray}\label{eq:D3.5}
x^{2}y + y + 1=0.
\end{eqnarray}
Multiplying equation (\ref{eq:D3.5}) by $y$ and then adding (\ref{eq:D3.4''}) gives
$$y=1,$$
this means $x=1$. However $f(1)=1$, it is a contradiction. So $f(x)=0$ if and only if $x=0$.

When $a\neq0$, we will show that $f(x)=a$ has one unique nonzero solution. Let $y=x^{2^{k}}$, $z=y^{2^{k}}$, $b=a^{2^{k}}$, $c=b^{2^{k}}$,
then $x^{2}=z^{2^{k}}$ and $a^{2}=c^{2^{k}}$. Hence we obtain
\begin{equation}\label{eq:D3.6}
\left\{
\begin{array}{l}
      xyz + xz + x =a,  \\
      x^{2}yz + x^{2}y + y =b,   \\
      x^{2}y^{2}z + y^{2}z + z=c.
\end{array}
\right.
\end{equation}
Eliminating the indeterminate $z$ in (\ref{eq:D3.6}), we deduce
\begin{numcases}{}
      (x^{2}+1)y^{2} + (ax+b+1)y + b =0,   \label{eq:D3.7}\\
      (x^{4}+1)y^{3} +(bx^{2}+b)y^{2}+ (cx^{2}+x^{2}+1)y + b =0. \label{eq:D3.7'}
\end{numcases}
Combining equations (\ref{eq:D3.7}) and (\ref{eq:D3.7'}) results in
\begin{equation}\label{eq:D3.8}
(ax^{3}+x^{2}+ax+1)y^{2} + (bx^{2}+cx^{2}+x^{2}+b+1)y + b =0.
\end{equation}
By equations (\ref{eq:D3.7}) and (\ref{eq:D3.8}), there is
\begin{equation}\label{eq:D3.9}
(x^{2} +1)y = a^{-1}(bx+cx+x+a).
\end{equation}

Obviously, $f(1)=1$. We claim that $x=1$ if and only if $a=1$. Suppose that $a=1$ and $x\neq1$, then we get that $b=c=1$ and $y=\frac{1}{x+1}$ from (\ref{eq:D3.9}). Plugging $y=\frac{1}{x+1}$ into equation (\ref{eq:D3.7}), we obtain $\frac{x}{x+1}=0$, which implies $x=0$. This contradicts the first assumption $x\neq0$.

Next, we consider the case of $a\neq0,1$. Plugging $y=\frac{(b+c+1)x+a}{a(x^{2} +1)}$ into equation (\ref{eq:D3.7}) leads to
\begin{equation}\label{eq:D3.10}
(a^{2}c+a^{2}+b^{2}+c^{2}+1)x +(a^{3}+ab^{2}+abc+ac+a) =0.
\end{equation}
If
\begin{equation}\label{eq:D3.11}
a^{2}c+a^{2}+b^{2}+c^{2}+1=0,
\end{equation}
raising both sides of (\ref{eq:D3.11}) to the $2^{k}$-th power gives
\begin{equation}\label{eq:D3.12}
a^{2}b^{2}+a^{4}+b^{2}+c^{2}+1=0.
\end{equation}
Adding (\ref{eq:D3.11}) and (\ref{eq:D3.12}), and then dividing $a^{2}$, there is
\begin{equation}\label{eq:D3.13}
a^{2}+b^{2}+c+1=0.
\end{equation}
Since $a\neq0$, we deduce that $a^{3}+ab^{2}+abc+ac+a=a(a^{2}+b^{2}+bc+c+1)=abc\neq0$, which implies that equation (\ref{eq:D3.10}) has no solution.
If $a^{2}c+a^{2}+b^{2}+c^{2}+1\neq0$, we have a solution of (\ref{eq:D3.10}).
Therefore, equation (\ref{eq:D3.10}) has at most one solution. The proof is completed.
\end{proof}

\section{A new class of permutation trinomials over $\mathbb{F}_{2^n}$ with $n\equiv 0({\rm mod}\ 4)$}\label{three}
In the section, one class of permutation trinomials is presented over $\mathbb{F}_{2^{n}}$ when $n\equiv 0({\rm mod}\ 4)$.

\begin{theorem}\label{th:d-c 1}
Let $k$ be an odd integer and let $m$, $n$, $d$ be positive integers satisfying $n=4m$, $1\leq k\leq n-1$, $\gcd(m, k)=1$, 
and $d=\sum_{i=0}^{2m}2^{ik}$. Then
\begin{eqnarray*}
f(x)=x^{d} + x^{2^{2m}} + x
\end{eqnarray*}
is a permutation polynomial on $\mathbb{F}_{2^n}$.
\end{theorem}

\begin{proof} For any fixed $a \in\mathbb{F}_{2^n}$, we need to prove that equation
\begin{eqnarray}\label{eq:d-c1.1}
f(x)=x^{d} + x^{2^{2m}} + x =a
\end{eqnarray}
has at most one solution in $\mathbb{F}_{2^n}$. Obviously, there is $f(x)=0$ when $x=0$. On the contrary, assume that $x\neq0$ is a solution of the equation
\begin{eqnarray}\label{eq:d-c1.2}
x^{d} + x^{2^{2m}} + x =0.
\end{eqnarray}
Raising $2^{2m}$-th power in both sides of (\ref{eq:d-c1.2}),  we get
\begin{eqnarray}\label{eq:d-c1.3}
x^{2^{2m}\cdot d} + x^{2^{2m}} + x =0.
\end{eqnarray}
Adding (\ref{eq:d-c1.2}) and (\ref{eq:d-c1.3}) gives
\begin{eqnarray}\label{eq:d-c1.4}
x^{2^{2m}\cdot d} + x^{d} =x^{d}(x^{(2^{2m}-1)\cdot d} + 1)=0.
\end{eqnarray}
It leads to $x=0$ or $x^{(2^{2m}-1)\cdot d}=1$. Since $\gcd(d, 2^{n}-1)=1$, this yields $x^{2^{2m}}=x$.
Plugging $x^{2^{2m}}=x$ into equation (\ref{eq:d-c1.2}), we have $x^{d}=0$, thereby $x=0$, which is impossible.
So $f(x)=0$ only has a solution $x=0$.

Since $k$ is an odd integer with $\gcd(m, k)=1$, it can be easily check that
$$d=\sum_{i=0}^{2m}2^{ik}=\frac{2^{(2m+1)k}-1}{2^{k}-1},$$
$\gcd(d, 2^{n}-1)=1$, and $\gcd(2^{k}-1, 2^{4m}-1)=1$.

In the sequel, assume that $x, a \in\mathbb{F}^{*}_{2^n}$. We need to show that equation
\begin{eqnarray}\label{eq:d-c1.5}
x^{d} + x^{2^{2m}} + x =a
\end{eqnarray}
has at most one solution. The above equation can be written as
\begin{eqnarray}\label{eq:d-c1.6}
x^{\frac{2^{(2m+1)k}-1}{2^{k}-1}} = x^{2^{2m}} + x+ a.
\end{eqnarray}
Raising $(2^{k}-1)$-th power in both sides of equation (\ref{eq:d-c1.6}), we derive that
\begin{eqnarray}\label{eq:d-c1.7}
x^{2^{(2m+1)k}-1} = (x^{2^{2m}} + x+ a)^{2^{k}-1}.
\end{eqnarray}
Let $x^{2^{2m}}=\theta x$, $\theta\neq0$, then $\theta^{2^{2m}+1}=1$ and
$$x^{2^{(2m+1)k}-1}=x^{2^{2m+k}-1}=(\theta x)^{2^{k}}x^{-1}=\theta^{2^{k}}x^{2^{k}-1}.$$
From equation (\ref{eq:d-c1.7}), we obtain
\begin{eqnarray}\label{eq:d-c1.8}
(\theta^{\frac{2^{k}}{2^{k}-1}}x)^{2^{k}-1} = ((1+\theta)x+ a)^{2^{k}-1}.
\end{eqnarray}
Since $\gcd(2^{k}-1, 2^{4m}-1)=1$, we deduce that $\theta^{\frac{2^{k}}{2^{k}-1}}x=(1+\theta)x+ a$, which can be written as $(1+\theta+\theta^{\frac{2^{k}}{2^{k}-1}})x=a$.
Because of $a\neq0$, we have
\begin{eqnarray}\label{eq:d-c1.9}
x=\frac{a}{1+\theta+\theta^{\frac{2^{k}}{2^{k}-1}}}.
\end{eqnarray}
Plugging (\ref{eq:d-c1.9}) into $x^{2^{2m}}=\theta x$, we get
$$\frac{a^{2^{2m}}}{1+\theta^{-1}+\theta^{-\frac{2^{k}}{2^{k}-1}}}=\theta \frac{a}{1+\theta+\theta^{\frac{2^{k}}{2^{k}-1}}},$$
which can be simplified as
$$a^{2^{2m}}(1+\theta+\theta^{\frac{2^{k}}{2^{k}-1}}) = a(1+\theta+\theta^{-\frac{1}{2^{k}-1}}).$$
Let $\theta^{\frac{1}{2^{k}-1}}=\beta$, we have $\beta^{2^{2m}+1}=1$. The above equation can be reduced as
$$a^{2^{2m}}(1+\beta^{2^{k}-1}+\beta^{2^{k}}) = a(1+\beta^{2^{k}-1}+\beta^{-1}).$$
It leads to
$$a^{2^{2m}}(\beta+\beta^{2^{k}}+\beta^{2^{k}+1}) = a(\beta+\beta^{2^{k}}+1),$$
which is equivalent to
\begin{eqnarray}\label{eq:d-c2.0}
a^{2^{2m}}\beta^{2^{k}+1} + (a+a^{2^{2m}})(\beta+\beta^{2^{k}})+a=0.
\end{eqnarray}
Set $\beta=t+w$, where $w^{3}=1$ and $w\neq1$. Since $k$ is odd, this implies $2^{k}\equiv 2({\rm mod}\ 3)$.
Thus equation (\ref{eq:d-c2.0}) can be written as
$$ a^{2^{2m}}(t+w)^{2^{k}+1} + (a+a^{2^{2m}})(t^{2^{k}}+t+w^{2^{k}}+w)+a=0,$$
which implies that
\begin{eqnarray}\label{eq:d-c2.1}
a^{2^{2m}}t^{2^{k}+1} + (w^{2}a^{2^{2m}}+a)t^{2^{k}}+(wa^{2^{2m}}+a)t=0.
\end{eqnarray}
If $t=0$, then $\beta=w$, and $\beta^{2^{2m}+1}=w^{2}\neq 1$, which is a contradiction.
Dividing equation (\ref{eq:d-c2.1}) by $t^{2^{k}+1}$, we obtain
\begin{eqnarray}\label{eq:d-c2.2}
(wa^{2^{2m}}+a)t^{-2^{k}} + (w^{2}a^{2^{2m}}+a)t^{-1}+a^{2^{2m}}=0.
\end{eqnarray}
Let $z=t^{-1}$, then $z\neq0$, equation (\ref{eq:d-c2.2}) becomes
\begin{eqnarray}\label{eq:d-c2.3}
(wa^{2^{2m}}+a)z^{2^{k}} + (w^{2}a^{2^{2m}}+a)z+a^{2^{2m}}=0.
\end{eqnarray}

If $wa^{2^{2m}}+a=0$ and $w^{2}a^{2^{2m}}+a=0$, then $a=0$. This is a contradiction.

If $wa^{2^{2m}}+a=0$ or $w^{2}a^{2^{2m}}+a=0$, then equation (\ref{eq:d-c2.3}) has one solution $z$.
Therefore equation (\ref{eq:d-c1.9}) only has one solution.

If $wa^{2^{2m}}+a\neq0$ and $w^{2}a^{2^{2m}}+a\neq 0$, assume that $z_{1}$, $z_{2}$ are two solutions of equation (\ref{eq:d-c2.3}),
then
\begin{eqnarray}\label{eq:d-c2.4}
(wa^{2^{2m}}+a)(z_{1}-z_{2})^{2^{k}} + (w^{2}a^{2^{2m}}+a)(z_{1}-z_{2}) =0.
\end{eqnarray}
Since $\gcd(2^{k}-1, 2^{4m}-1)=1$, we obtain $z_{1}=z_{2}$ or $z_{1}=z_{2}+\big(\frac{w^{2}a^{2^{2m}}+a}{wa^{2^{2m}}+a} \big)^{\frac{1}{2^{k}-1}}$ from (\ref{eq:d-c2.4}).
Therefore, equation (\ref{eq:d-c2.3}) has at most two solutions. It is not difficult to check that $z=1$ is a solution of equation (\ref{eq:d-c2.3}).
Then $t=1$, $\beta=t+w=w^{2}$, which implies $\beta^{2^{2m}+1}=w\neq 1$. It contradicts the known result $\beta^{2^{2m}+1}=1$. So equation (\ref{eq:d-c2.3}) only has one solution $z$
such that equation (\ref{eq:d-c1.9}) has at most one solution.
This completes the proof.

\end{proof}

\section{Conclusion}\label{conclusion}
In this paper, we propose six new classes of permutation trinomials with explicit forms over
$\mathbb{F}_{2^n}$ by determining the solutions of some equations. The multivariate method
introduced by Dobbertin is a powerful tool in this work.
We expect that the new permutation trinomials can be applied to design new linear codes or cryptography.

\end{document}